\documentclass[10pt,reqno]{amsart}
\usepackage{amsmath}
\usepackage{amsfonts}
\usepackage{amsthm}
\usepackage{amsbsy}
\usepackage{amssymb}
\usepackage{epsfig}
\usepackage{psfrag}
\usepackage{latexsym}

\newcommand\beque{\begin{equation*}}
\newcommand\beq{\begin{equation}}
\newcommand\eeq{\end{equation}}
\newcommand\eeque{\end{equation*}}
\newcommand\eqn{\begin{eqnarray}}
\newcommand\beqna{\begin{eqnarray*}}
\newcommand\eeqna{\end{eqnarray*}}
\newcommand\feqn{\end{eqnarray}}
\newcommand{\dirac}{\slash \!\!\!\partial}

\newtheorem{lemma}{Lemma}

\numberwithin{equation}{section}

\begin{document}

\date{\today}

\title[On a polynomial zeta function]{On a polynomial zeta function}

\author{S. L. Cacciatori}
\address{Dipartimento di Fisica e Matematica, Universit\`a dell'Insubria, 22100
Como, Italy, and\\
I.N.F.N., sezione di Milano, Italy}
\email{sergio.cacciatori@uninsubria.it}

\begin{abstract}

We introduce a polynomial zeta function $\zeta^{(p)}_{P_n}$, related to certain problems of mathematical physics,
and compute its value and the value of its first derivative at the origin $s=0$, by means of a very simple technique.
As an application, we compute the determinant of the Dirac operator on quaternionic vector spaces.

\end{abstract}

\maketitle

\section{Introduction}
\label{intro}

The aim of this paper is to look at few simple properties of a particular class of zeta functions, which we dub
{\it polynomial zeta functions} because they are associated with a polynomial $P_n$ of degree $n$:
$$
\pmb{\zeta}^{(p)}_{P_n} (s)=\sum_{k=0}^\infty \frac {k^p}{P_n (k)^s}.
$$
Our motivations arise from physical questions, where, substantially,
$\pmb{\zeta}^{(p)\prime}_{P_n} (0)$ defines the determinant of an operator ${\mathcal O}$ having $\lambda_k=P_n(k)$ as eigenvalues each 
one having degeneration $k^p$:
$$
\det {\mathcal O} =e^{-\pmb{\zeta}^{(p)\prime}_{P_n} (0)}
$$
When ${\mathcal O}$ represent the Hamiltonian operator of a field theory, this provides the (Euclidean) effective
action
$$
S_{\rm eff} =-\log (\det {\mathcal O})=\pmb{\zeta}^{(p)\prime}_{P_n} (0).
$$
Thus, exact calculations of non perturbative quantum field theories effects can be performed, as for example the Schwinger effect
on flat or curved background, see for example \cite{blauvisswi}, \cite{barzer}, or discharge effects on certain charged black hole
backgrounds, \cite{belcapia1}, \cite{belcapia2}. \\
However, we think that a general investigation of the properties of the polynomial zeta functions could be of a certain
interest for pure mathematics also.\\
Our strategy is very simple and is essentially based on a direct application of the Abel-Plana formula.

\section{The polynomial zeta function}
Let us consider the polynomial
\eqn
P_n(x)=a_n x^n +a_{n-1} x^{n-1}+\ldots+a_1 x +a_0
\feqn
with $a_0 a_n \neq 0$.
Suppose that the zeros of $P_n$ are not in $\mathbb{N}$. We define the {\it polynomial zeta function}
\eqn
\pmb{\zeta}^{(p)}_{P_n} (s):= \sum_{k=0}^\infty \frac {k^p}{P_n(k)^s}, \qquad\ p\in \mathbb{N}.
\feqn
We are interested to compute $\pmb{\zeta}^{(p)}_{P_n} (0)$ and $\pmb{\zeta}^{(p)\prime}_{P_n} (0)$. \\
To this hand, let us first assume that all zeros
$x_i$, $i=1,\ldots,n$, of $P_n$ have negative real part: $\Re x_i<0$. With these hypothesis, the function
$$
f(z)=\frac {z^p}{P_n(z)^s}
$$
for $\Re s>0$ is regular on $\Re z\geq 0$ and
satisfies the conditions
\begin{enumerate}
\item $\lim_{y\to \infty} e^{-2\pi |y|} f(x+iy) =0,$
\item $\lim_{x\to +\infty} \int_{-\infty}^{+\infty} e^{-2\pi |y|} f(x+iy) dy =0.$
\end{enumerate}
So, we can apply the Abel-Plana formula \cite{whitwat} to get
\eqn
\pmb{\zeta}^{(p)}_{P_n} (s)= \frac 1{2a_0^s} \delta_{p,0} +\int_{0}^\infty \frac {x^p}{P_n(x)^s} dx+ i \int_0^\infty \left[
\frac {(ix)^p}{P_n(ix)^s} -\frac {(-ix)^p}{P_n(-ix)^s} \right] \frac 1{e^{2\pi x}-1} dx.  \label{zeta}
\feqn
We write the polynomial $P_n$ as
$$
P_n(x)=a_n \prod_{i=1}^n (x-x_i)\equiv a_n P_n^0 (x),
$$
so that $P_n^0$ is monic. We also define
$$
\tilde {P}_n (x)=x^n P_n^0 (1/x)=\prod_{i=1}^n (1-x_i x).
$$
To perform our computations we need two lemmas.
\begin{lemma} \label{lemma1}
Let $P^0_n $ be a monic polynomial of degree $n$ and having all zeros $x_i$, $i=1,\ldots, n$ with negative real part, and $p\in \mathbb{N}$. Then
\eqn
\lim_{s\to 0} \int_0^\infty x^p \left\{ \frac 1{P^0_n (x)^s} -\frac 1n \sum_{i=1}^n \frac 1{(x-x_i)^s} \right\}dx=0.
\feqn
\end{lemma}
Here we mean the limit of the analytic continuation $f(s)$ of the expression defined by the integral.
\begin{proof}
We have
$$
f(s)=\int_0^1 x^p \left\{ \frac 1{P^0_n (x)^s} -\frac 1n \sum_{i=1}^n \frac 1{(x-x_i)^s} \right\}dx+
\int_1^\infty x^p \left\{ \frac 1{P^0_n (x)^s} -\frac 1n \sum_{i=1}^n \frac 1{(x-x_i)^s} \right\}dx.
$$
For the first integral we can exchange the limit with the integral, so that it vanishes when $s\mapsto0$. For the second integral, it is
convenient to take the change of integration variable $x\mapsto 1/x$ so that it becomes
$$
I(s)=\int_0^1 \left[ \frac {x^{ns-p-2}}{\tilde {P}_n (x)^s} -\frac 1n \sum_{i=1}^n \frac {x^{s-p-2}}{(1-x_i x)^s} \right]dx.
$$
The integrand has a branch point in $x=0$, so that we introduce a cut on the positive real semiaxis. Next, we define a path $\Gamma$, starting from
the point $x=1\equiv e^{0}$, going around $x=0$ and then ending to the point $x=1^+\equiv e^{2\pi i}$. See the figure 1.
\begin{figure}[h]\label{fig}
\begin{center}
\includegraphics[angle=90,scale=0.30]{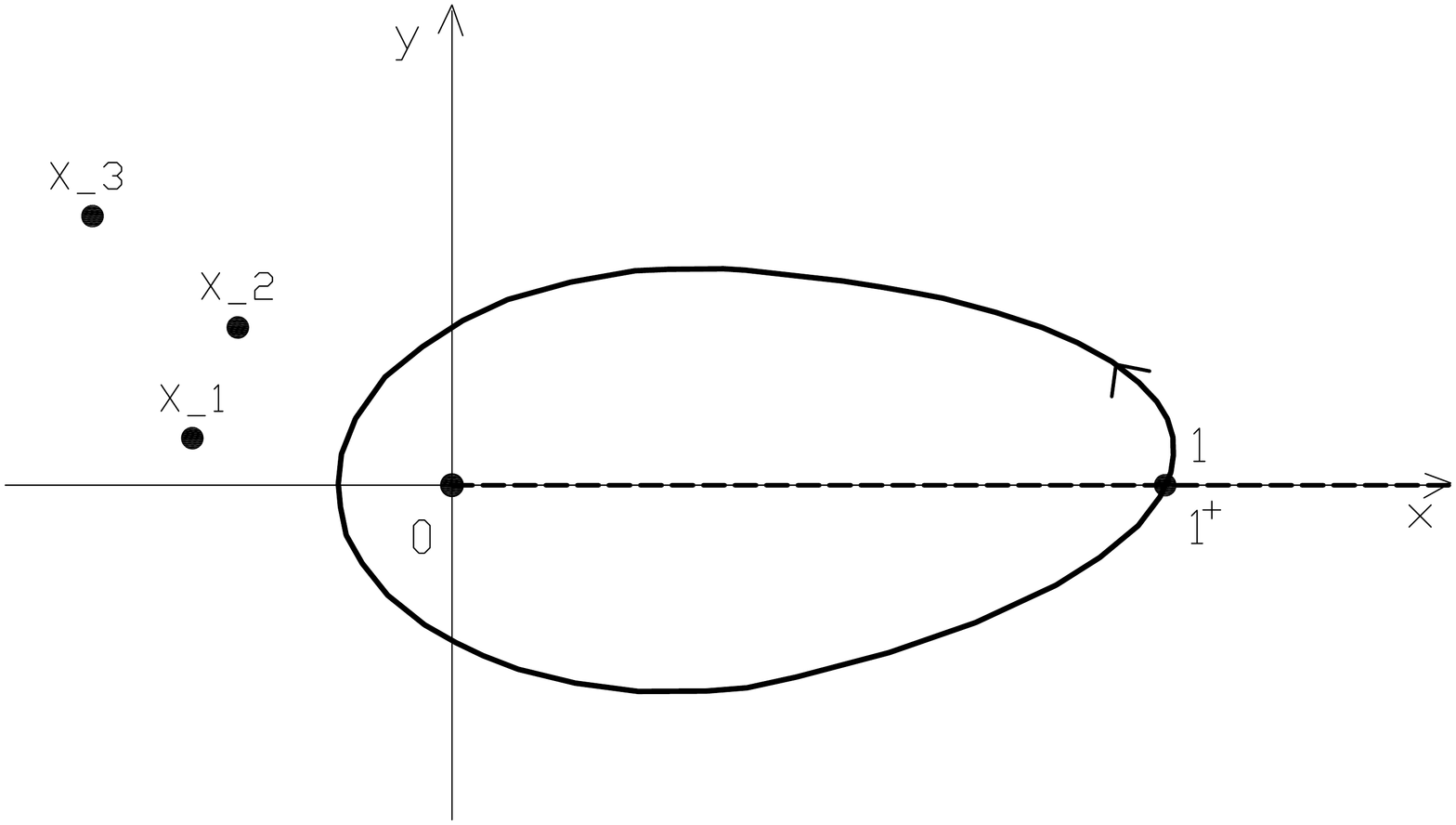}
\caption{\em The path $\Gamma$.}
\end{center}
\end{figure}
This is not a closed path, but in the usual way
it is easy to show that, when $I$ is well defined, we obtain
$$
I(s)=\frac 1{e^{2\pi i n s}-1} \int_\Gamma \frac {x^{ns-p-2}}{\tilde {P}_n (x)^s} dx -
\frac 1{e^{2\pi i s}-1} \frac 1n \int_\Gamma  \sum_{i=1}^n \frac {x^{s-p-2}}{(1-x_i x)^s} dx.
$$
As $\Gamma$ does not pass through $x=0$, this integral are well defined in a neighborhood of $s=0$ and we can take the limit
or differentiate under the integral sign. Note that for $s=0$ the monodromy becomes trivial and the path closes, so that both integrals
vanish. Thus we can use the de l'Hospital rule to get
$$
\lim_{s\to0} I(s)=\frac 1{2\pi i} \int_\Gamma \frac 1{x^{p+2}} \left[ (\log x-\frac 1n \log \tilde {P}_n (x)) -\frac 1n
(n\log x -\sum_{i=1}^n \log (1-x_i x))  \right]dx =0.
$$
\end{proof}
A second important step is the following:
\begin{lemma} \label{lemma2}
Let $p$ a positive integer and $F^{(p)}_{P^0_n}(s)$ the analytic continuation on the complex plane of the function
$$
F^{(p)}_{P^0_n}(s)=\int_0^{\infty} x^p  \left\{ \frac {\log P^0_n (x)}{P^0_n(x)^s}
-\sum_{i=1}^n \frac {\log (x-x_i)}{(x-x_i)^s}   \right\}dx,
\quad\
\Re (s)> p+1.
$$
Then
\eqn
F^{(p)}_{P^0_n}(0)=-\frac 1{2n} \sum_{l=1}^p \left[ \frac 1{l(p+1-l)} (\sum_{i=1}^n x_i^l) (\sum_{j=1}^n x_j^{p-l+1})\right]
+\frac {H_p}{p+1}  \sum_{i=1}^n x_i^{p+1},
\feqn
where $H_p =\sum_{i=1}^p \frac 1i$ are the harmonic numbers. In particular $F^{(0)}_{P^0_n}(0)=0$.
\end{lemma}
\begin{proof}
The proof is very similar to the previous one so that we will be essential. As before we first obtain
\begin{eqnarray*}
&& \lim_{s\to0} F^{(p)}_{P^0_n}(s)=\lim_{s\to0} [J(s)-\sum_{i=1}^n J_i(s)], \cr
&& J(s):= \int_0^1 \left[ x^{ns-p-2} \frac {\log \tilde {P}_n (x)-n\log x}{\tilde {P}_n(x)^s} \right]dx,\cr
&& J_i(s):= \int_0^1 \left[ x^{s-p-2} \frac {\log (1-x_i x)-\log x}{(1-x_i x)^s} \right]dx
\end{eqnarray*}
Takeing $\Gamma$ as above, and noting the monodromy of the logarithm, we have
$$
\int_\Gamma \left[ x^{ns-p-2} \frac {\log \tilde {P}_n (x)-n\log x}{\tilde {P}_n(x)^s} \right]dx=
(e^{2\pi i n s}-1) J(s) -\frac {2\pi i n e^{2\pi i n s} }{e^{2\pi i n s}-1} \int_\Gamma\frac {x^{ns-p-2}}{\tilde {P}_n (x)^s}dx
$$
so that
$$
J(s)=\frac 1{e^{2\pi i n s}-1} \int_\Gamma \left[ x^{ns-p-2} \frac {\log \tilde {P}_n (x)-n\log x}{\tilde {P}_n(x)^s} \right]dx
+\frac {2\pi i n e^{2\pi i n s} }{(e^{2\pi i n s}-1)^2}\int_\Gamma\frac {x^{ns-p-2}}{\tilde {P}_n (x)^s}dx
$$
and similarly
$$
J_i(s)=\frac 1{e^{2\pi i s}-1} \int_\Gamma \left[ x^{s-p-2} \frac {\log (1-x_i x)-\log x}{(1-x_i x)^s} \right]dx
+\frac {2\pi i e^{2\pi i s} }{(e^{2\pi i s}-1)^2}\int_\Gamma\frac {x^{s-p-2}}{(1-x_i x)^s}dx.
$$
Again, we can use the de l'Hospital rule to obtain
$$
F^{(p)}_{P^0_n}(0)=-\frac 12 \frac 1{(p+1)!} \frac {d^{p+1}}{dx^{p+1}} \left[ \frac 1n \left(\sum_{i=1}^n \log (1-x x_i)\right)^2
-\sum_{i=1}^n \left( \log (1-x x_i)\right)^2 \right]_{x=0}.
$$
Note that for $p=0$ this vanishes. For $p\geq 1$, using the Leibnitz rule for derivation, we obtain
\begin{eqnarray*}
&& F^{(p)}_{P^0_n}(0)=-\frac 1{2(p+1)!} \sum_{l=1}^p \binom {p+1}{l} \left[ \frac 1n
\left( \sum_{i=1}^n \frac {d^{l-1}}{dx^{l-1}} \frac {x_i}{1-x_i x} \right)
\left( \sum_{j=1}^n \frac {d^{p-l}}{dx^{p-l}} \frac {x_j}{1-x_j x} \right)\right. \cr
&& {\phantom {F^{(p)}_{P^0_n}(0)=}}
\left.-\sum_{i=1}^n \left( \frac {d^{l-1}}{dx^{l-1}} \frac {x_i}{1-x_i x} \frac {d^{p-l}}{dx^{p-l}} \frac {x_i}{1-x_i x}\right)
\right]\cr
&& {\phantom {F^{(p)}_{P^0_n}(0)}} =
-\frac 12 \sum_{l=1}^p \left[ \frac 1{l(p-l+1)} \left( \frac 1n (\sum_{i=1}^n x_i^l) (\sum_{j=1}^n x_j^{p-l+1}) -
\sum_{i=1}^n x_i^{p+1} \right) \right],
\end{eqnarray*}
from which the thesis follows.
\end{proof}

\subsection{\boldmath{$\pmb{\zeta}_{P_n}^{(2m)}(0)$}}
From (\ref{zeta}) and using Lemma \ref{lemma1}, if $p=2m$, we get
$$
\pmb{\zeta}^{(2m)}_{P_n} (0)=\frac 12 \delta_{m,0}+\lim_{s\to 0}\int_0^{\infty} \frac {x^p}{P_n(x)^s} dx=
\frac 12 \delta_{m,0}+\lim_{s\to 0} \frac 1n \sum_{i=1}^n \int_0^{\infty} \frac {x^p}{(x-x_i)^s} dx.
$$
Using
\eqn
\int_0^\infty \frac {x^p}{(x-x_i)^s}dx=\frac {\Gamma (s-p-1) p!}{\Gamma (s) } (-x_i)^{p+1-s} \label{relaz}
\feqn
we get
\eqn
\pmb{\zeta}^{(2m)}_{P_n} (0) =\frac 12 \delta_{m,0}- \frac 1{2m+1} \sum_{i=1}^{n} (-x_i)^{2m+1}. \label{zeta0pari}
\feqn
For example, using $x_1+\ldots +x_n =-a_{n-1}/a_n$, we get
\eqn
\pmb{\zeta}_{P_n} (0)\equiv \pmb{\zeta}^{(0)}_{P_n} (0)=\frac 12 -\frac 1n \frac {a_{n-1}}{a_n}.
\feqn

\subsection{\boldmath{$\pmb{\zeta}_{P_n}^{(2m+1)}(0)$}}
For the case $p=2m+1$, we get
$$
\pmb{\zeta}^{(2m+1)}_{P_n} (0)=\lim_{s\to 0}\int_0^{\infty} \frac {x^p}{P_n(x)^s} dx -(-1)^m 2\int_0^\infty \frac {x^{2m+1}}{e^{2\pi x}-1}dx.
$$
From Lemma \ref{lemma1} and (\ref{relaz}) we get
\eqn
\pmb{\zeta}^{(2m+1)}_{P_n} (0) = \frac 1{2m+2} \sum_{i=1}^{n} x_i^{2m+2} -(-1)^m 2 \frac {(2m+1)!}{(2\pi)^{2m+2}} \zeta(2m).\label{zeta0dispari}
\feqn
Using
$$
\sum_{i=1}^n x_i^2 =\frac {a_{n-1}^2}{a_n^2} -2 \frac {a_{n-2}}{a_n},
$$
we find for example
\eqn
\pmb{\zeta}^{(1)}_{P_n} (0)= \frac 12 \frac {a_{n-1}^2}{a_n^2}-\frac {a_{n-2}}{a_n}-\frac 1{12}.
\feqn

\subsection{\boldmath{$\pmb{\zeta}_{P_n}^{(p)\prime}(0)$}}
We first note that
\eqn
\pmb{\zeta}^{(p)\prime}_{P_n} (0)=- \pmb{\zeta}^{(p)}_{P_n} (0) \log a_n + \pmb{\zeta}^{(p)\prime}_{P^0_n} (0).\label{trucchetto}
\feqn
Differentiating (\ref{zeta}) with respect to $s$, for $P_n^0$, we obtain
$$
\pmb{\zeta}^{(p)\prime}_{P^0_n} (s)=-\frac 12 \frac {\sum_{i=1}^n \log (-x_i) }{\prod_{i=1}^n (-x_i)^s} \delta_{p,0}
-\int_0^\infty x^p \left\{ \frac {\log P^0_n (x)}{P^0_n(x)^s}
+i \left[ i^p\frac {\log P^0_n(ix)}{P^0_n(ix)^s} -(-i)^p\frac {\log P^0_n(-ix)}{P^0_n(-ix)^s} \right] \frac 1{e^{2\pi x}-1}
\right\}dx.\label{dezeta}
$$
Now
\eqn
&& \lim_{s\to 0} i\int_0^\infty \left[ (ix)^p\frac {\log P^0_n(ix)}{P^0_n(ix)^s} -(-ix)^p\frac {\log P^0_n(-ix)}{P^0_n(-ix)^s} \right]
\frac {dx}{e^{2\pi x}-1} \cr
&& \qquad\ =i\sum_{i=1}^n \int_0^\infty [(ix)^p\log (ix-x_i)-(-ix)^p\log (-ix-x_i)]\frac {dx}{e^{2\pi x}-1}\cr
&& \qquad\ =-i\frac d{ds}_{|_{s=0}} \sum_{i=1}^n \int_0^\infty [(ix)^p(ix-x_i)^{-s}-(-ix)^p(-ix-x_i)^{-s}]  \frac {dx}{e^{2\pi x}-1}.\nonumber
\feqn
We can apply the Abel-Plana formula to this expression to obtain
$$
\pmb{\zeta}^{(p)}_{P^0_n} (0)=\lim_{s\to0} \int_0^\infty x^p \left\{ \frac {\log P^0_n (x)}{P^0_n(x)^s}-
\sum_{i=1}^n \frac {\log (x-x_i)}{(x-x_i)^s}   \right\} dx +\sum_{i=0}^n {\zeta}^{(p)\prime}_H (-x_i,0),
$$
where
$$
{\zeta}^{(p)}_H (a,s)=\sum_{k=0}^{\infty} \frac {k^p}{(k+a)^s}.
$$
Writing
$$
k^p= [(k+a)-a]^p=\sum_{l=0}^p \binom {p}{l} (k+a)^l (-a)^{p-l},
$$
we see that
$$
{\zeta}^{(p)}_H (a,s)=\sum_{l=0}^p \binom {p}{l} (-a)^{p-l} \zeta_H (a,s-l)
$$
where
$$
\zeta_H (a,s)=\sum_{n=0}^{\infty} \frac 1{(n+a)^s}
$$
is the classical Hurwitz zeta function \cite{whitwat}.
From Lemma \ref{lemma2} and (\ref{trucchetto}) we finally get
\begin{eqnarray}
&& \pmb{\zeta}^{(p)\prime}_{P_n} (0)=- \pmb{\zeta}^{(p)}_{P_n} (0) \log a_n +\frac 1{2n} \sum_{l=1}^p \left[ \frac 1{l(p+1-l)} (\sum_{i=1}^n x_i^l)
(\sum_{j=1}^n x_j^{p-l+1})\right] -\frac {H_p}{p+1}  \sum_{i=1}^n x_i^{p+1} \cr
&& \phantom{\pmb{\zeta}^{(p)\prime}_{P_n} (0)=} +\sum_{i=1}^n \sum_{l=0}^p \binom {p}{l} x_i^{p-l} \zeta'_H (-x_i, -l). \label{zetaprimo0}
\end{eqnarray}
For example, for $p=0$ and $p=1$ we find
\eqn
&& \pmb{\zeta}^{\prime}_{P_n} (0)\equiv \pmb{\zeta}^{(0)\prime}_{P_n} (0)=-\left( \frac 12 -\frac 1n \frac {a_{n-1}}{a_n} \right)\log a_n
+\log \frac {\prod_{i=1}^n \Gamma(-x_i)}{(2\pi)^{\frac n2}}, \\
&& \pmb{\zeta}^{(1)\prime}_{P_n} (0)= -\left( \frac 12 \frac {a^2_{n-1}}{a^2_n} -\frac {a_{n-2}}{a_n} -\frac 1{12}
\right)\log a_n +\frac 12 \left[ \frac {a^2_{n-1}}{a^2_n} \left( \frac 1n -1 \right) +2 \frac {a_{n-2}}{a_n} \right] \cr
&& \phantom{\pmb{\zeta}^{(1)\prime}_{P_n} (0)=} +\sum_{i=1}^n \left( x_i \log \frac {\Gamma (-x_i)}{\sqrt {2\pi}} +\zeta'_H (-x_i, -1) \right),
\feqn
where we used the identity
$$
\zeta'_H (a,0) =\log \frac {\Gamma (a)}{\sqrt {2\pi}}.
$$

\

\subsection{Remark} We determined formulas (\ref{zeta0pari}), (\ref{zeta0dispari}) and (\ref{zetaprimo0}) with the assumption that
all zeros of the polynomial have negative real part. However, it is easy to see that the same formulas hold true for the general case, the only
assumption being $x_i \notin \mathbb{N}$ for any $i=1,\ldots,n$. To prove this it suffices to note that
\eqn
\pmb{\zeta}^{(p)}_{P_n} (s)= \sum_{k=0}^{N-1} \frac {k^p}{P_n(k)^s}+ \sum_{k=0}^\infty \frac {(k+N)^p}{P_n(k+N)^s}
=\sum_{k=0}^{N-1} \frac {k^p}{P_n(k)^s}+\sum_{l=0}^p \binom {p}{l} N^{p-l} \pmb{\zeta}^{(l)}_{P_n^{(N)}} (s) \label{final}
\feqn
where $P_n^{(N)} (x):=P_n(x+N)$ has zeros in $x_i-N$. If we choose $N>{\rm Max}_{i=1,\ldots,n} \{\Re x_i\}$, then we can differentiate
(\ref{final}) and use our previous results to $\pmb{\zeta}^{(l)}_{P_n^{(N)}} (s)$ to show that they hold true for $\pmb{\zeta}^{(p)}_{P_n} (s)$ also.

\

\section{Final Remarks}
We studied the polynomial zeta function $\pmb{\zeta}^{(p)}_{P_n}$ limiting ourself to the problem of compute its value and the value of its first
derivative at $s=0$. This is because our interest in physical applications. For example, our results can be used to compute the Dirac operator
on a quaternionic projective space $\mathbb {HP}^n$. Its eigenvalues are \cite{milhorat}
$$
\pm \sqrt {\lambda_{m,k}}, \qquad \pm \sqrt {\mu_{m,k}}
$$
with
\eqn
&& \lambda_{m,k} = 4m^2 +4m(2n+k) +4n(n+k)-4(k+1), \qquad\ k\in \{1,2,\ldots,n-1 \}, \ m\in \mathbb {N}^0, \cr
&& \mu_{m,k}=\lambda_{m,k+1}+8(k+1), \qquad\ k\in \{1,2,\ldots,n-1 \},\ m\in\left\{ \begin{array}{c} \mathbb{N} \quad \mbox{if $k=n$} \\
\mathbb{N}^0 \quad \mbox{otherwise} \end{array} \right. . \nonumber
\feqn
The logarithm of the determinant of the Dirac operator is
$$
\zeta'_{\dirac} (0) =\frac 12 \zeta'_{\dirac^2} (0) -\frac 14 \zeta_{\dirac} (0) \log (-1).
$$
Thus we can apply our results to get
$$
\det (\dirac)_{\mathbb {HP}^n} =e^{\zeta'_{\dirac} (0)}=(-1)^{\frac n8 (5n+3)}
\frac {\pi^{2n-1}}{2^{5n^2-2n-2} (n-2)! [(n-1)!(n+1)!]^{n-1} (2n)!} \prod_{k=1}^{n-1} \frac 1{(n+k-1)! (n+k+1)!}.
$$
Another very interesting physical application is to the problem of describing the discharging of a charged black hole. In the case
when this is described by a Nariai solution, both the Dirac and the Klein-Gordon equations can be explicitly solved and
the effective action describing the Schwinger pair production can be computed exactly. This phenomenon will be described in details
elsewhere \cite{belcapia1}, \cite{belcapia1},
so that we only note that, in that case, it arises the problem to compute the derivative in $s=0$ of multiple sums of the form
(see also \cite{elizalde}, \cite{sprea2})
\eqn
\tilde{\pmb{\zeta}}^{(m)}_{P_n}(s)=\sum_{\vec m \in \mathbb{Z}^m} \frac 1{P_n(m_1+\ldots+m_m)^s}=\sum_{k=0}^\infty \frac {d^{(m)}_k}{P_n(k)^s},
\feqn
where $d^{(m)}_k$ is the cardinality of the set of partitions $k=m_1+\ldots+m_m$ of $k$. Now, using the recurrence relation
$d^{(m)}_n=\sum_{l=0}^n d^{(m-1)}_l$ one can easily expand $d^{(m)}_k$ as a polynomial
$$
d^{(m)}_k=\sum_{l=0}^m c_l k^l.
$$
For example
$$
d^{(2)}_k=k+1, \qquad\ d^{(3)}_k=\frac 12 k^2+\frac 32 k+1, \qquad\ d^{(4)}_k =\frac 16 k^3 +k^2 +\frac {11}6 k +1.
$$
Then
\eqn
\tilde{\pmb{\zeta}}^{(m)\prime}_{P_n}(0)=\sum_{l=0}^m c_l \pmb{\zeta}^{(l)\prime}_{P_n} (0).
\feqn
Finally, going beyond the physical applications, we think that the polynomial zeta function is interesting by itself so that its
properties deserve a deeper investigation \cite{sprea1}, \cite{sprea3}.

\section*{{\bf Acknowledgments}}
I am grateful to F. Belgiorno both for suggestions and encouragement and to F. Dalla Piazza for useful discussions.
I am indebted with R. Marigo for the figure. Finally, I thank S. Pigola for useful comments.


\end{document}